\documentclass[10pt,english,twocolumn,conference]{IEEEtran}
\usepackage{amsmath}
\usepackage{lipsum}
\usepackage{amsthm}
\usepackage{float}
\usepackage{mathrsfs}
\usepackage{verbatim}
\usepackage{epstopdf}
\usepackage{amssymb}
\usepackage{amsfonts}
\usepackage{subfigure}
\usepackage{color}
\usepackage{cite}
\usepackage{cancel}
\usepackage{amsmath,amssymb,amsthm,mathrsfs,amsfonts,dsfont}
\usepackage[shortlabels]{enumitem}

\newcommand\independent{\protect\mathpalette{\protect\independenT}{\perp}}
\def\independenT#1#2{\mathrel{\rlap{$#1#2$}\mkern2mu{#1#2}}}

\usepackage[utf8]{inputenc}
\usepackage{graphicx}

\theoremstyle{plain}
\newtheorem{thm}{Theorem} 

\newtheorem{Prop}{Proposition}
\newtheorem{Rem}{Remark}
\newtheorem{corr}{Corollary}

\theoremstyle{definition}
\newtheorem{defn}{Definition}

\title{Privacy Against Brute-Force Inference Attacks}

\author{\IEEEauthorblockN{Seyed Ali Osia\IEEEauthorrefmark{1},
Borzoo Rassouli\IEEEauthorrefmark{2}, Hamed Haddadi\IEEEauthorrefmark{3}, Hamid R. Rabiee\IEEEauthorrefmark{1}, Deniz G\"{u}nd\"{u}z \IEEEauthorrefmark{3}}
\IEEEauthorblockA{\IEEEauthorrefmark{1} Sharif University of Technology,  osia@ce.sharif.edu, rabiee@sharif.edu}
\IEEEauthorblockA{ \IEEEauthorrefmark{2} University of Essex, b.rassouli@essex.ac.uk}
\IEEEauthorblockA{ \IEEEauthorrefmark{3} Imperial College London, \{h.haddadi, d.gunduz\}@imperial.ac.uk}
}

\begin{document}

\maketitle
\begin{abstract}
Privacy-preserving data release is about disclosing information about useful data while retaining the privacy of sensitive data. Assuming that the sensitive data is threatened by a brute-force adversary, we define \emph{Guessing Leakage} as a measure of privacy, based on the concept of \emph{guessing}. After investigating the properties of this measure, we derive the optimal utility-privacy trade-off via a linear program with any $f$-information adopted as the utility measure, and show that the optimal utility is a concave and piece-wise linear function of the privacy-leakage budget.
\end{abstract}
\section{Introduction}

Large-scale data collection and analysis is a key component of many recent technological advances such as autonomous driving, online health monitoring, reliable energy grid, and intelligent IoT systems. While these data-driven applications provide better services by efficiently processing user data in massive scales, collection and sharing of personal data is increasingly creating privacy risks, especially with the advances in machine learning and data mining algorithms. Thus, developing privacy-preserving data release mechanisms that jointly consider the utility from shared data with the associated privacy leakage is key to the wide scale adoption of some of these emerging technologies. 
As a classic example, publishing a general purpose database, or even releasing its aggregate statistics, may threaten an individual's privacy, which has led to the introduction of widely used techniques such as 
k-anonymity \cite{sweeney2002}, and differential privacy \cite{dwork06}. 


The information-theoretic formulation of privacy provides a general statistical framework to model both the utility and privacy, and allows investigating their trade-off as an optimization problem. Modeling the available dataset with random variable (r.v.) $Y$  and the private/sensitive latent variable with $X$, the data that should be released, denoted by $U$ is obtained as the output of a privacy-preserving statistical kernel (transformation), and can be obtained as the solution of an optimization problem. 
Mutual information is the most common measure of both utility and privacy, whereby the trade-off becomes the \emph{privacy funnel} \cite{makhdoumi2014}. 
However, there is no optimal/universal definition of privacy and it can either abstract away from the adversarial threats  \cite{sweeney2002}, \cite{Borzoo}, or depend on the vulnerability of the sensitive data to adversarial attacks~\cite{Liao, asoodeh2018}. In this paper, we assume that the adversary is a brute-force attacker, i.e., it performs an exhaustive trial and error attack over all the possible realizations of the sensitive data. 
There are many examples of such brute-force attacks in real life, where criminals steal private information by determining a cipher key via an exhaustive search \cite{christiansen2015}, or checking several potential shortened URLs to discover active links \cite{chu2013security}.


We assume that a brute-force adversary performs a number of guesses to successfully determine the value of private data $X$, modeled as a r.v. with finite support. A privacy-preserving mechanism is built against this brute-force adversary. To this end, we define \textit{guessing leakage}, and investigate its properties as a privacy measure. Afterwards, we formulate the utility-privacy trade-off as a non-convex optimization problem, and derive the optimal data release mechanism via a linear program (LP).

The problem of guessing is a well-established area in information theory. In \cite{Massey}, Massey proposed a lower bound on the minimum expected number of guesses needed to find $X$ in terms of its Shannon entropy $H(X)$, while in \cite{Arikan}, different moments of the guessing function are lower bounded in terms of the Renyi entropy of $X$, and the result is used to analyze the computational complexity of sequential decoding. Later works, such as  \cite{hanawal2011guessing}, apply large deviation techniques to more general scenarios in this context.
The problem of guessing also appears in Shannon-theoretic cryptography in \cite{merhav1999shannon}, i.e., encryption against a brute-force wiretapper. The problem of computational security against a guessing attacker has been addressed in \cite{christiansen2015, beirami2015quantifying}. The practical challenges of guessing passwords is studied in \cite{kelley2012guess,ur2015measuring}. 
However, limited attention has been devoted to the effect of data sharing on guessability of private information from a privacy-preserving point of view. In \cite{rachlin2009}, a sub-optimal utility-privacy trade-off is found by using the lower bound provided in \cite{Arikan}. 


\textbf{Notations.} R.v.'s are denoted by capital letters, and their realizations by lower case letters. Matrices and vectors are denoted by bold capital and bold lower case letters, respectively. 
For a positive integer $n$, we define $[n]\triangleq\{1,2,\ldots,n\}$. 
For a finite alphabet $\mathcal{X}$, the probability simplex $\mathcal{P}(\mathcal{X})$ is the standard $(|\mathcal{X}|-1)$-simplex. 
Furthermore, to each probability mass function (pmf) $p_{X}(\cdot)$ corresponds a probability vector $\mathbf{p}_X\in \mathcal{P}(\mathcal{X})$, whose  $i$-th element is $p_X(x_i)$ ($i\in[|\mathcal{X}|]$). Likewise, for a pair of r.v.'s $(X,Y)$ with joint pmf $p_{X,Y}$, the probability vector $\mathbf{p}_{X|y}$ corresponds to the conditional pmf $p_{X|Y}(\cdot|y),\forall y\in\mathcal{Y}$, and $\mathbf{P}_{X|Y}$ is an $|\mathcal{X}|\times|\mathcal{Y}|$ matrix with columns $\mathbf{p}_{X|y},\forall y\in\mathcal{Y}$. Statistical independence between $X$ and $Y$ is shown as $X\independent Y$.
For a convex function $f$ such that $f(1)=0$, and the probability mass functions $p,q$ on $\mathcal{X}$, the $f$-divergence is defined as\footnote{We assume that $p$ is absolutely continuous with respect to $q$, i.e., $q(x)=0$ implies $p(x)=0$.} $D_f(p||q)\triangleq \sum_{x\in\mathcal{X}}q(x)f(\frac{p(x)}{q(x)})$. Finally, for a pair $(X,Y)\sim p_{X,Y}$, the $f$-information is defined as $I_f(X;Y)\triangleq D_f(p_{X,Y}||p_X\cdot p_Y)$.



\section{System model}
\subsection{Preliminaries}

Consider a triplet of discrete r.v.'s $(X,Y,W)\in\mathcal{X}\times\mathcal{Y}\times\mathcal{W}$, with finite alphabets, and distributed according to $p_{X,Y,W}$. Let $Y$ and $X$ denote the useful data to be revealed, and the sensitive data to be concealed, respectively. As in \cite{Ishwar}, $W$ denotes what the user/curator directly observes, which may be a noisy representation of the pair $(X,Y)$. Assume that the \textit{privacy mapping}/\textit{data release mechanism} takes $W$ as input and maps it to the \textit{released data} denoted by $U$. In this scenario, $(X,Y)-W-U$ form a Markov chain, and the privacy mapping is captured by the conditional distribution $p_{U|W}$. 

The aim of the privacy mapping is to simultaneously preserve the fidelity of $Y$ and the privacy of $X$ by the release of $U$.
In this paper, we measure the utility of the release by the $f$-information between $Y$ and $U$, i.e., $I_f(Y;U)$, which incorporates mutual information as a special case, and define the privacy measure in the sequel.

\subsection{Guessing}
Consider the problem of guessing the realization of a discrete r.v. $X\in\mathcal{X}$ by asking questions of the form "Is $X$ equal to $x$?", until the answer is "Yes." As in \cite{Arikan}, a \textit{guessing function/strategy} of $X$ is denoted by a bijection $G(X):\mathcal{X}\to[|\mathcal{X}|]$. Hence, for a given guessing strategy $G(\cdot)$, $G(x)$ represents the number of required guesses when $X=x$. Likewise, for a pair of r.v.'s $(X,Y)$, $G(X|Y):\mathcal{X}\times\mathcal{Y}\to[|\mathcal{X}|]$ is the guessing function of $X$ given $Y$, i.e., $G(x|y)$ denotes the number of guesses required to determine $X=x$, when $Y=y$. For a given pmf $p_X$, let $G^*$ denote the minimizer(s)\footnote{It can be readily verified that $G^*$ is unique if and only if $p_X(x_i)\neq p_X(x_j),\ \forall i\neq j.$} of $\mathds{E}[G(X)]$, which guesses the value of $X$ in decreasing order of probabilities as shown in \cite{Massey}. The guessing entropy is defined as $H_G(X)\triangleq\mathds{E}[G^*(X)]$, which, from the guesser's point of view, measures the uncertainty in $X$, as it denotes the minimum average number of guesses required to determine its realization. We have $H_G(X)\in[1,\frac{|\mathcal{X}|+1}{2}]$, where the minimum or maximum are attained when $X$ is, respectively, deterministic or uniformly distributed over $\mathcal{X}$. Throughout the paper, the guessing entropy $H_G(X)$ and $H_G(\mathbf{p}_X)$ are written interchangeably\footnote{It is noteworthy to mention that $H_G(X)$ as the minimum number of guesses is not consistent with intuition when $|\mathcal{X}|=1$ or $2$, as in the former no guessing is needed, while $H_G(X)=1$, and in the latter, one guess is sufficient to determine the value of $X$, while $H_G(X)>1$. Although a more exact notion would be $H_G(X)\triangleq \min\{\mathds{E}[G^*(X)],|\mathcal{X}|-1\}$, one can always justify $H_G(X)$ as the minimum number of guesses by emphasizing on the need of receiving an affirmative answer, i.e., until the answer is "Yes". Having said that, the analysis of this paper remains valid in either cases.}.

\begin{defn}
   For a probability vector $\mathbf{p}\in\mathcal{P}(\mathcal{X})$, Let the \emph{rank vector} of $\mathbf{p}$, denoted by $\mathbf{r}_{\mathbf{p}}$, be a vector that labels the elements of $[|\mathcal{X}|]$ according to their order induced by sorting $\mathbf{p}$ in descending order. For example, for $\mathbf{p}_1=\begin{bmatrix}0.6&0.1&0.3\end{bmatrix}^T$, we have $\mathbf{r}_{\mathbf{p}_1}=\begin{bmatrix}1&3&2\end{bmatrix}^T$. For $\mathbf{p}_2=\begin{bmatrix}0.5&0.25&0.25\end{bmatrix}^T$, either of $\begin{bmatrix}1&2&3\end{bmatrix}^T$ or $\begin{bmatrix}1&3&2\end{bmatrix}^T$ could be a candidate\footnote{It can be verified that by imposing the constraint $r_i<r_j$ ($i\neq j$) if and only if $p_i-p_j+\mathds{1}\{p_i=p_j\}(j-i)>0$, i.e., equal probabilities being ranked based on their index order, the rank vector is well defined. Nonetheless, the (possible) ambiguity in the definition of $\mathbf{r}$ is not problematic here.} for $\mathbf{r}_{\mathbf{p}_2}$.  
\end{defn}

\begin{defn}\label{partition}
    The \emph{rank partition} $\mathcal{P}_\mathbf{r}(\mathcal{X})$ is the set of all probability vectors $\mathbf{p}\in\mathcal{P}(\mathcal{X})$ with rank vector $\mathbf{r}$, where $\mathbf{r}$ is an arbitrary permutation of the elements in $[|\mathcal{X}|]$. Hence, the probability simplex $\mathcal{P}(\mathcal{X})$ can be divided into $|\mathcal{X}|!$ equal-rank partitions.
\end{defn}



\begin{Prop}\label{Prop1}
The guessing entropy $H_G(X)$ is a piece-wise linear and concave functional (see Figure \ref{fig:g3d}) of $p_X(\cdot)$.
\end{Prop}
\begin{proof}
     Let $X\sim\mathbf{p}$, and it is immediate that $H_G(X)=\mathbf{r}_{\mathbf{p}}^T\cdot\mathbf{p}$, which proves the piece-wise linearity of $H_G(X)$ in $\mathbf{p}$. Furthermore, we have $\mathbf{r}_{\mathbf{p}}^T\cdot\mathbf{p}\leq\mathbf{r}^T\cdot\mathbf{p}$ for any rank vector $\mathbf{r}\neq\mathbf{r}_{\mathbf{p}}$, which follows from \cite{Massey}. Hence, the concavity is proved as follows. For $\lambda\in[0,1]$ and two arbitrary $\mathbf{p},\mathbf{q}\in\mathcal{P}(\mathcal{X})$, let $\tilde{\mathbf{p}}\triangleq \lambda\mathbf{p}+(1-\lambda)\mathbf{q}$. We have,
    \begin{align*}
       H_G(\tilde{\mathbf{p}})&=\mathbf{r}_{\tilde{\mathbf{p}}}^T\cdot\tilde{\mathbf{p}}\\
        &=\lambda\mathbf{r}_{\tilde{\mathbf{p}}}^T\cdot\mathbf{p}+(1-\lambda)\mathbf{r}_{\tilde{\mathbf{p}}}^T\cdot\mathbf{q}\\
        &\geq \lambda\mathbf{r}_{{\mathbf{p}}}^T\cdot\mathbf{p}+(1-\lambda)\mathbf{r}_{{\mathbf{q}}}^T\cdot\mathbf{q}\\
        &=\lambda H_G({{\mathbf{p}}})+(1-\lambda)H_G({{\mathbf{q}}}),
    \end{align*}
    where the inequality is strict if and only if $\mathbf{p},\mathbf{q}$ belong to different rank partitions. 
   \end{proof}

\begin{figure}
    \centering
    \includegraphics[width=.8\columnwidth]{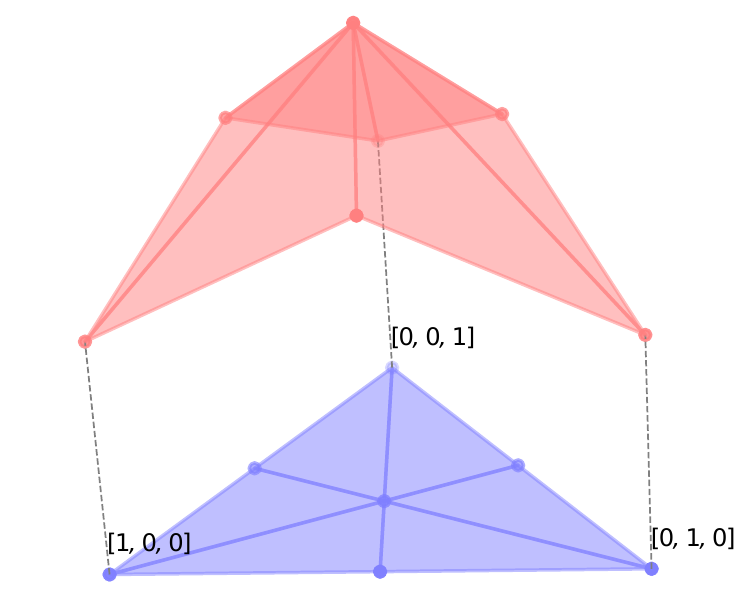}
    \caption{Guessing entropy for a ternary r.v. $X$. The probability simplex $\mathcal{P}(\mathcal{X})$ is divided into $|\mathcal{X}|!=6$ rank partitions.}
    \label{fig:g3d}
\end{figure}

\subsection{Brute-force inference attack}

A brute-force adversary aims at inferring the private data, i.e., $X$, via trial and error, or equivalently, guessing. Here, we adopt a model similar to the cost function-based inference in \cite{du2012} by defining the cost of inference attack as the number of trials which lead to a correct guess.
Prior to observing any realization of the released data $U$, the adversary uses the optimal guessing strategy $G^*(\cdot)$ to minimize its cost of inference, i.e., $C^*_0 = H_G(X)$. After observing $U=u$, the adversary can update its belief about the private data as the posterior $p_{X|U}(\cdot|u)$, which in turn leads to an updated\footnote{Note that knowledge of the posterior is not necessary in the sense that the adversary could also get the same updated strategy by only knowing the rank vector associated to it, i.e., $\mathbf{r}_{\mathbf{p}_{X|u}}$.} guessing strategy $G^*(\cdot|u)$. Therefore, the minimum cost of inference is now $C^*_u = H_G(X|U=u)$. Considering the average additive gain (where the average is over $U$) that the adversary sees in its inference cost, we are ready to define the privacy leakage measure and investigate some of its properties as follows.

\begin{defn}The \textit{guessing leakage} is defined as:
\begin{equation}\label{GL}
    GL(X\longrightarrow U)\triangleq H_G(X)-H_G(X|U).
\end{equation}
\end{defn}
\begin{Prop}
We have
\begin{equation}\label{GLpositive}
    GL(X\longrightarrow U)\geq 0,
\end{equation}
with equality if and only if the rank vector associated with $p_{X|U}(\cdot|u)$, i.e., $\mathbf{r}_{\mathbf{p}_{X|u}}$, does not change with $u, \; \forall u \in \mathcal{U}$. Note that this is a weaker condition than statistical independence.
\end{Prop}
\begin{proof}
    From Proposition \ref{Prop1}, we have
    \begin{align*}
      H_G(X|U)&=\sum p(u)H_G(\mathbf{p}_{X|u})  \\
      &\leq H_G\left(\sum p(u)\mathbf{p}_{X|u}\right)\\
      &=H_G(X).
    \end{align*}
    where the inequality is tight if and only if all the $\mathbf{p}_{X|u}$'s  ($\forall u$) belong to the same rank partition.
    Alternatively, (\ref{GLpositive}) can be proved by \cite[Corollary 1]{Arikan}.
\end{proof}

\begin{Rem}
The difference in (\ref{GL}) is reminiscent of the mutual information expanded in terms of entropies, which follows from its logarithmic nature. Nonetheless, one can observe that $GL(X\longrightarrow U)$ has no relation with mutual information in terms of one being a lower/upper bound to the other. Let $X,U$ be binary random variables with $X|\{U=u_i\}\sim \mbox{Bern}(\frac{i}{i+3}),\ i=1,2$. It is immediate that $X\not\independent U$, and hence, $I(X;U)>0$. Also, since we have $\mathbf{r}_{\mathbf{p}_{X|u_1}}=\mathbf{r}_{\mathbf{p}_{X|u_2}}$, we get $GL(X\longrightarrow U)=0$. Therefore, $I(X;U)>GL(X\longrightarrow U)$. In another example, let $X=U$, which results in $I(X;U)=H(X)$, and $GL(X\longrightarrow U)=H_G(X)-1$. According to \cite[Section III]{Massey}, we can have distributions for which $H(X)$ is vanishingly small, while $H_G(X)$ is large. As a result, $I(X;U)<GL(X\longrightarrow U)$.
\end{Rem}

A privacy measure can be investigated in terms of data processing inequalities. 
\begin{defn}
	A privacy measure $J(X;U)$ is said to satisfy the post-processing inequality \cite{Ishwar} if and only if for any Markov chain $X-U-U'$, we have $J(X;U')\leq J(X;U)$, and the linkage inequality, if and only if $J(X;U') \leq J(U;U')$.
\end{defn}

\begin{thm} Guessing leakage satisfies the post-processing inequality, but not the linkage inequality.
\end{thm}
\begin{proof}
	Consider the Markov chain $X-U-U'$. We have
	\begin{align}
		H_G(X|U') &= \sum_{u'}p(u')\nonumber	H_G(X|U'=u')\nonumber\\
		&= \sum_{u'}p(u')H_G\left(\sum_u p(u|u')\mathbf{p}_{X|u}\right)\nonumber\\
		&\geq \sum_{u'}p(u')\sum_u p(u|u')H_G(\mathbf{p}_{X|u})\label{eq1}\\
		& = \sum_{u}p(u)H_G(\mathbf{p}_{X|u})\nonumber\\
		& = H_G(X|U), \nonumber
	\end{align}
	where (\ref{eq1}) follows from the concavity of $H_G$. Hence, we have
	\begin{equation*}
	    GL(X\longrightarrow U')\leq GL(X\longrightarrow U).
	\end{equation*}
	In other words, no independent processing of the released data, $U$, can increase the privacy leakage. 
	
	To show that $GL$ does not satisfy the linkage inequality, let $U'\sim\mbox{Bern}(\theta)$ for some $\theta\in(0,1)$. Also, let $\mathcal{U}=[3]$ with $\mathbf{p}_{U|u'_0}=[.35\ .4\ .25]^T$ and $\mathbf{p}_{U|u'_1}=[.3\ .6\ .1]^T$. By setting $X=U\mbox{ mod }2$, we have $X-U-U'$ form a Markov chain, and $GL(X\longrightarrow U')>GL(U\longrightarrow U')=0$.
\end{proof}

\section{Utility-privacy trade-off }
Having defined the utility and privacy measures, the utility-privacy trade-off can be written as\footnote{We assume that $Y\not\independent W$, since otherwise, $t(\epsilon)=0,\ \forall \epsilon.$}:
\begin{equation}\label{toff1.5}
t(\epsilon)\triangleq\sup_{\substack{p_{U|W}:\\(X,Y)-W-U\\GL(X\longrightarrow U)\leq\epsilon}}I_f(Y;U),
\end{equation}
where $\epsilon \in [0,GL(X\longrightarrow W)]$. By the routine application of cardinality bounding techniques \cite{Elgamal}, it becomes sufficient to have $|\mathcal{U}|\leq|\mathcal{W}|+1$. Furthermore, the supremum can be replaced by maximum, since a continuous objective function attains its supremum over a compact set. 
Also, searching over $p_{U|W}$ can be equivalent\footnote{Trivially, for any pair $(p_U,\mathbf{p}_{W|u})$, $(X,Y)-W-U$ can be constructed.} to searching over the pair $(p_U,\mathbf{p}_{W|u})$, such that $\mathbf{p}_W=\sum_u p_U(u)\mathbf{p}_{W|u}$. Therefore, the problem reduces to

\begin{equation}\label{toff2}
t(\epsilon)=\!\!\!\!\!\!\!\!\!\!\!\!\!\!\!\!\max_{\substack{p(\cdot),\mathbf{p}_{W|u}\in\mathcal{P}(\mathcal{W}):\\\sum p(u)\mathbf{p}_{W|u}=\mathbf{p}_W,\\\sum p(u)H_G\left(\mathbf{P}_{X|W}\mathbf{p}_{W|u}\right)\geq H_G(X)-\epsilon}}\!\!\!\!\!\!\!\!\!\!\!\!\!\!\!\!\sum_u p(u)D_f\left(\mathbf{P}_{Y|W}\mathbf{p}_{W|u}||\mathbf{p}_{Y}\right).
\end{equation}

For an arbitrary rank vector $\mathbf{r}$, let $\mathcal{Q}_\mathbf{r}(\mathcal{W})$ denote the inverse image of $\mathcal{P}_\mathbf{r}(\mathcal{X})$ under $\mathbf{P}_{X|W}$, which is a linear transformation from $\mathcal{P}(\mathcal{W})$ to $\mathcal{P}(\mathcal{X})$. For a given $\mathbf{r}$, $\mathcal{Q}_\mathbf{r}(\mathcal{W})$ is a convex polytope with a finite number of extreme points, since it can be written as the intersection of a finite number of closed half-spaces in $\mathcal{P}(\mathcal{W})$\footnote{Depending on the transformation $\mathbf{P}_{X|W}$, $\mathcal{Q}_\mathbf{r}(\mathcal{W})$ could be an empty set for some values of $\mathbf{r}$. Note that the empty set can be viewed as a null polytope.}. For example, with $\mathbf{r}=[1\ 3\ 2]^T$, we have $\mathcal{Q}_\mathbf{r}(\mathcal{W})=\{\mathbf{p}\in\mathcal{P}(\mathcal{W})|\mathbf{v}_1\cdot\mathbf{p}\geq\mathbf{v}_3\cdot\mathbf{p},\ \mathbf{v}_3\cdot\mathbf{p}\geq\mathbf{v}_2\cdot\mathbf{p}\}$, which forms the intersection of $\mathcal{P}(\mathcal{W})$ with two closed half-spaces, where $\mathbf{v}_i$ denotes the $i$th row of $\mathbf{P}_{X|W}$. Let $\mathbb{Q}_r$ denote the set of extreme points of $\mathcal{Q}_\mathbf{r}(\mathcal{W})$. As a result, any element of $\mathcal{Q}_\mathbf{r}(\mathcal{W})$ can be written as a convex combination of the elements of $\mathbb{Q}_r$. 

In what follows, we show that there is no loss of optimality in replacing $\mathbf{p}_{W|u}\in\mathcal{P}(\mathcal{W})$ in (\ref{toff2}) with $\mathbf{p}_{W|u}\in\mathbb{Q}$, where $\mathbb{Q}\triangleq\cup_r\mathbb{Q}_r$. It is already known that the $f$-divergence, $D_f(p||q)$, is convex in $(p,q)$, and in $p$ for a fixed $q$, which implies the convexity of the objective function of (\ref{toff2}) in $\mathbf{p}_{W|u}$. For an arbitrary $\mathbf{p}_{W|u}\in\mathcal{P}(\mathcal{W})$, we have $\mathbf{p}_{W|u}\in\mathcal{Q}_\mathbf{r}(\mathcal{W})$ for some $\mathbf{r}$. Writing this $\mathbf{p}_{W|u}$ as a convex combination of the elements of $\mathbb{Q}_r$ does not alter $H_G(\cdot)$, which follows from the piece-wise linearity of $H_G(\cdot)$, i.e., $H_G\left(\mathbf{P}_{X|W}\mathbf{z}\right)$ is linear in $\mathbf{z}$ for $\mathbf{z}\in\mathcal{Q}_\mathbf{r}(\mathcal{W})$. Furthermore, this does not decrease the objective function, which is a direct consequence of the convexity of $D_f\left(\mathbf{P}_{Y|W}\mathbf{z}||\mathbf{p}_{Y}\right)$ in $\mathbf{z}$. Therefore, in (\ref{toff2}), $\mathbf{p}_{W|u}\in\mathcal{P}(\mathcal{W})$ can be replaced with $\mathbf{p}_{W|u}\in\mathbb{Q}$, which leads to the following theorem.

\begin{thm}\label{thm:lp}
	The utility-privacy trade-off in (\ref{toff2}) can be solved by a linear program (LP).
\end{thm}
\begin{proof}
	Let $\{\mathbf{q}_i\}_{i=1}^{k}$ denote the elements of $\mathbb{Q}$. The problem reduces to
	\begin{align}\label{lp}
	\begin{split}
	t(\epsilon)=\max_{p(\cdot)\geq 0} & \sum_{i=1}^k p(u_i)\  D_f\left(\mathbf{P}_{Y|W}\mathbf{q}_{i}||\mathbf{p}_{Y}\right)\\
	\text{s.t.} & \sum_{i=1}^k p(u_i)\  \textbf{r}_{(\mathbf{P}_{X|W}\mathbf{q}_i)}^T\mathbf{P}_{X|W}\mathbf{q}_i  \geq H_G(X)-\epsilon, \\
	& \sum_{i=1}^k p(u_i) \mathbf{q}_i = \mathbf{p}_W, 
	\end{split}
	\end{align}
	which is an LP\footnote{Note that the constraint $\sum_i p(u_i)=1$ is redundant, as it is implicitly implied by $\sum_{i}p(u_i)\mathbf{q}_i = \mathbf{p}_W$.}.
\end{proof}


\begin{corr}
The utility-privacy trade-off in (\ref{toff1.5}) is a concave and piece-wise linear function of $\epsilon$ (see Fig.~\ref{fig:opt}).
\end{corr}
\begin{proof}
    This follows from the LP sensitivity analysis \cite[Lemma 2]{jansen1997sensitivity}.
\end{proof}
\begin{Rem}
If the elements of $\mathbf{p}_X$ are distinct, it is always possible to reveal some information with no privacy leakage, i.e., $t(0)>0$. Moreover, if all the columns of $\mathbf{P}_{X|W}$ belong to the same rank partition, no data release can cause any leakage of privacy, i.e., $t(\epsilon)=I_f(W;W),\forall \epsilon.$
\begin{proof}
    If the elements of $\mathbf{p}_X$ are distinct, we have $\mathbf{p}_X\in\mbox{int}(\mathcal{P}_{\mathbf{r}}(\mathcal{X}))$ for some $\mathbf{r}$. Since $Y$ and $W$ are not independent, there exists a vector $\mathbf{v}$ such that i) $\mathbf{p}_Y\neq\mathbf{P}_{Y|W}(\mathbf{p}_W+\eta\mathbf{v})$ for sufficiently small $\eta>0$ and ii) $\mathbf{1}^T\cdot\mathbf{v}=0$. Let $U'\sim\mbox{Bern}(\frac{1}{2})$, and set $\mathbf{p}_{W|u'_i}=\mathbf{p}_W+(-1)^i\eta\mathbf{v}$ for $i=1,2.$ As a result, $p_{Y,U'}\neq p_Y\cdot p_{U'}$, which results in $I_f(Y;U')>0$. Also, for sufficiently small $\eta$, $\mathbf{p}_{X|u'_i}$ ($=\mathbf{P}_{X|W}\mathbf{p}_{W|u'_i}, i=1,2$) lie in the same rank partition as $\mathbf{p}_X$, which results in $GL(X\longrightarrow U')=0$. Therefore, we have $t(0)\geq I_f(Y;U')>0$ The second claim is proved by setting $U=W$, using the data-processing inequality for $f$-information, and noting that $GL(X\longrightarrow W)=0$.
\end{proof} 
\end{Rem}

\begin{Rem}
In order to specify the LP for a particular setting, it is sufficient to identify the elements of $\mathbb{Q}$. The procedure of finding the extreme points of a convex polytope is a classical problem,  which  is  omitted  here  due  to  lack  of space. In the special case of $W=X$, i.e., when the curator has direct access only to the sensitive data, these points are already known as $(1,0,\ldots,0)^T$, $(\frac{1}{2},\frac{1}{2},0,\ldots,0)^T$, $(\frac{1}{3},\frac{1}{3},\frac{1}{3},0,\ldots,0)^T$ and so on, and their permutations, which are nothing but the extreme points of all the rank partitions of $\mathcal{P}(\mathcal{W})$.
\end{Rem}

\begin{Rem}
Theorem 2 is a direct consequence of the convexity of the objective function, and the piece-wise linearity of the privacy measure. Hence, a similar analysis applies when the utility is measured by the minimum mean square error (MMSE), minimum probability of error, or $GL(Y\longrightarrow U)$.
\end{Rem}

\begin{Rem}
Guessing leakage is defined in (\ref{GL}) as the additive gain in the inference cost of a brute-force attacker. Alternatively, one could define the following multiplicative gain as the privacy-leakage measure:
\begin{equation}\label{GLm}
   GL_m(X\longrightarrow U)\triangleq \frac{H_G(X)}{H_G(X|U)},
\end{equation}
which belongs to $[1,\frac{|\mathcal{X}|+1}{2}]$. It can be readily verified that the two theorems of this paper remain valid when (\ref{GLm}) replaces (\ref{GL}). When the adversary is a memoryless guesser, i.e., each new guess is independent of the previous guesses, we have
\begin{equation*}
   \log_2GL_m(X\longrightarrow U)={I_{\frac{1}{2}}(X;U)},
\end{equation*}
where $I_{\alpha}(X;U)$ denotes the Arimoto mutual information of order $\alpha(\geq 0)$, which follows from \cite[Theorem 1]{Huleihel}. However, this does not hold in general, and is only pertinent to the special case of a memoryless guesser. 
\end{Rem}

\section{Numerical Results}

\begin{figure}
    \centering
    \includegraphics[width=\columnwidth]{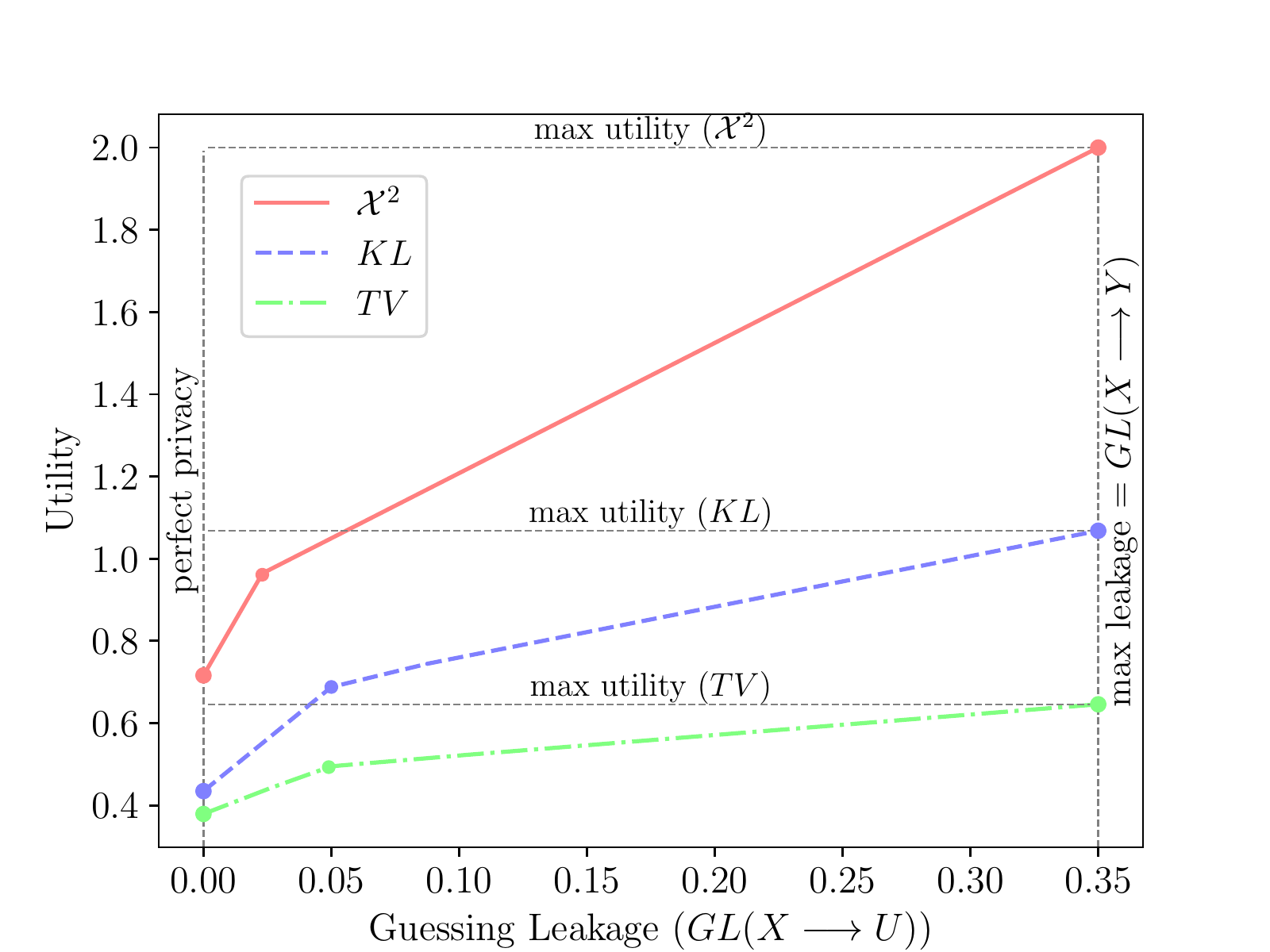}
    \caption{Optimal utility-privacy trade-off for three different utility measures.}
    \label{fig:opt}
\end{figure}

Suppose an attacker has access to a list of user IDs that potentially belong to individuals with highly sensitive job positions. The goal of the attacker is to discover whether any of these IDs are related to a person working for a list of presumed government agencies. Each organization has a registered domain name, hence the attacker can build all of the possible email addresses by combining user IDs with domain names, and exhaustively checking the existence of emails, which can potentially be used for consequent social engineering. Assuming that checking a large number of emails is a risky task for the attacker, they will try to minimize the number of attempts. Again, assuming the attacker has access to the location check-ins of users in a social network and also knows the exact location of all branches of the organizations, they can wisely reorder the emails and check them in a manner that minimizes the number of trials.

Based on this scenario, suppose we are interested in three organizations, $x_1$, $x_2$ and $x_3$, with 500, 300 and 200 employees, respectively. Let $X$ denote the organization variable with probability vector $\mathbf{p}_X=[0.5, 0.3, 0.2]$. Considering 100 user IDs and without any side information, the average number of trials is 170 ($=100*H_G(X)$). Now assume we have 3 cities, $y_1$, $y_2$ and $y_3$, and each organization has offices in 2 of them (for example, $x_1$ has 300 employees in $y_1$ and 200 employees in $y_2$, $x_2$ has 150 employees in $y_1$ and 150 employees in $y_3$, and $x_3$ has 60 employees in $y_2$ and 140 employees in $y_3$.). In this way, we have defined $P_{Y|X}$ where $Y$ denotes the city variable. If the attacker knows the city variable for all users, the minimum number of trials is $100*H_G(X|Y)=135$, which is less than 170 ($GL(X\longrightarrow Y)=0.35$). Hence, a privacy preserving mapping is needed for check-ins of the people with sensitive positions, which can be done by intentionally creating wrong check-ins. This is equivalent to sampling $U$ from $P_{U|Y}$, which obviously increases the level of privacy while it decreases the correctness of location check-ins (utility). Therefore, we should build a trade-off between utility and leakage. Figure~\ref{fig:opt} shows the optimal utility-privacy trade-off for this example, when the utility is measured by three variants of $f$-information corresponding to the $\mathcal{X}^2$-divergence, $KL$-divergence, and total variation ($TV$) distance.

\section{Conclusions and future work}

We considered the problem of privacy against brute-force adversaries. By investigating the properties of guessing entropy, we introduced \emph{guessing leakage} as a privacy measure. We studied the optimal utility-privacy trade-off with $f$-information as the utility measure, and showed it to be the solution of an LP. Unless the curator has direct access only to the private data, we need to identify the extreme points of a convex polytope, whose complexity grows exponentially. Hence, sub-optimal algorithms are to be sought, such as restricting the search space to the set of all deterministic mappings, which is the subject of our ongoing work. Another practical direction is to address this problem when the curator is uncertain (or even unaware) of the underlying distribution.

\bibliographystyle{IEEEtran}
\bibliography{IEEEabrv,ref}

\end{document}